\pdfoutput=1

\documentclass{sig-alternate-10}
\usepackage{algpseudocode,algorithm,algorithmicx}
\usepackage{amsbsy}
\usepackage{amssymb}
\usepackage{amsmath}
\usepackage{graphicx}
\usepackage{subcaption}
\usepackage{capt-of}
\usepackage{mathrsfs} 
\usepackage[usenames,dvipsnames]{color}
\usepackage[colorlinks=true,linkcolor=blue,citecolor=Mahogany]{hyperref}
\usepackage[poorman]{cleveref}
\usepackage[comma,numbers,square,sort]{natbib} 
\usepackage{multirow}
\usepackage{mathtools}
\usepackage{nicefrac}
\usepackage{relsize}
\usepackage{thmtools, thm-restate}

\usepackage{xspace}
\newcommand{\ignore}[1]{}
\usepackage{cleveref}
\renewcommand{\cref}{\Cref}
\newtheorem{proposition}{\bf Proposition}

\newtheorem{theorem}{\bf Theorem}
\newtheorem{lemma}{\bf Lemma}
\newtheorem{corollary}{\bf Corollary}
\newtheorem{definition}{\bf Definition}
\newtheorem{claim}{\bf Claim}

\newdef{example}{\bf Example}
\newtheorem{problem}{\bf Problem}

\newcommand{\AODAG}{AND-OR DAG }

\newcommand{\margG}{\textsc{MarginalGreedy }}
\newcommand\p{\ensuremath{\mathsf{P}}}
\newcommand{\R}{\mathbb{R}}
\newcommand{\bigoh}{\mathcal{O}}
\newcommand\np{\ensuremath{\mathsf{NP}}}

\usepackage{balance}  
\begin{document}


\title{Efficient and Provable Multi-Query Optimization}



%
%
%
%

\numberofauthors{2} 

\author{
%
%
\alignauthor
Tarun Kathuria\\
       \affaddr{Microsoft Research}\\
       \email{t-takat@microsoft.com}
\alignauthor
S. Sudarshan\\
       \affaddr{Indian Institute of Technology Bombay}\\
       \email{sudarsha@cse.iitb.ac.in}}

\maketitle

\begin{abstract}
Complex queries for massive data analysis jobs have become increasingly commonplace. Many such queries contain common subexpressions, either within a single query or among multiple queries submitted as a batch. Conventional query optimizers do not exploit these subexpressions and produce sub-optimal plans. The problem of multi-query optimization (MQO) is to generate an optimal \textit{combined} evaluation plan by computing common subexpressions once and reusing them. Exhaustive algorithms for MQO explore an $\bigoh(n^n)$ search space. Thus, this problem has primarily been tackled using various heuristic algorithms, without providing any theoretical guarantees on the quality of their solution.

In this paper, instead of the conventional cost minimization problem, we treat the problem as maximizing a linear transformation of the cost function. We propose a greedy algorithm for this transformed formulation of the problem, which under weak, intuitive assumptions, provides an approximation factor guarantee for this formulation. We go on to show that this factor is optimal, unless $\p = \np$. Another noteworthy point about our algorithm is that it can be easily incorporated into existing transformation-based optimizers. 
We finally propose optimizations which can be used to improve the efficiency of our algorithm.
\end{abstract}
\vspace{-0.5em}

\vspace{-0.5em}

\vspace{-0.5em}
\section{Introduction}
Modern data analytics platforms frequently have to run scripts which contain a large number of complex queries. Often, these queries contain common subexpressions due to the nature of the analysis performed. These subexpressions may occur within a single complex query which i) contains multiple correlated nested subqueries or ii) if the database contains many materialized views which are referenced multiple times in the query. A more interesting case where common subexpressions arise is when a batch of related queries are being executed together.

Conventional query optimizers are not suited for such scenarios since they do not exploit these subexpressions and instead produce locally optimal plans for each query.  These plans can be globally sub-optimal since they do not make use of the shared subexpressions while generating the plans. The goal of multi-query optimization (MQO) is to generate query plans where these subexpressions are executed once and their results used by multiple consumers. The best plan is selected in a completely cost-based manner.

We now present an example to illustrate the MQO problem and how locally optimal plans may be globally sub-optimal for multiple queries in the presence of common subexpressions.
\begin{example}\label{firstExample}(Example 1.1 in \cite{Roy2000})
Consider a batch consisting of two queries $(A \bowtie B \bowtie C)$ and $(B \bowtie C \bowtie D)$ whose locally optimal plans (i.e., individual best plans) are $(A \bowtie B) \bowtie C$ and $(B \bowtie C) \bowtie D$ respectively. The individual best plans for the two queries do not have any common subexpressions. However, consider a locally sub-optimal plan for the first query $A \bowtie (B \bowtie C)$. It is clear that $(B \bowtie C)$ is a common subexpression and can be computed once and used by both queries. 

Consider the following instantiation of the various costs for the two queries shown in Figure \ref{fig:mqoexample}. Suppose the base relations \textit{A, B, C} and \textit{D} each have a scan cost of 10 units. Each of the joins have a cost of 100 units, giving a total evaluation cost of 460 units for the locally optimal plans shown in Figure \ref{fig:mqoexample}a. On the other hand, in the plan shown in Figure \ref{fig:mqoexample}b, the common subexpression $(B \bowtie C)$ is first computed and materialized on the disk at a cost of 10. Then, it is scanned twice - the first time to join with A in order to compute the first query, and the second time to join it with D in order to compute the second - at a cost of 10 per scan. Each of these joins have a cost of 100 units. Thus, the total cost of this consolidated plan is 370 units, which is lesser than the cost of the locally optimal plan in Figure \ref{fig:mqoexample}a.

It should be noted that blindly sharing a subexpression may not always lead to a globally optimal strategy. For example, there may be cases where the cost of joining the subexpression $(B \bowtie C)$ with $A$ is very large compared to the cost of the plan $(A \bowtie B) \bowtie C$; in such cases it may make no sense to reuse $(B \bowtie C)$ even if it were available. \hfill $\square$
\end{example}

The benefits of a good algorithm for MQO are not just restricted to multiple queries in a batch but can also be used to find better plans for a single complex query. Consider an example of a large query consisting of multiple subqueries with a common subexpression between two subqueries. Traditional Volcano-style\cite{Graefe1993} transformation rules-based query optimizers will not consider such sharing unless it is explicitly stated as a transformation rule. Of course, one cannot state all such possible rules and, thus, these cases of sharing in are not considered by Volcano when an optimal plan is devised for single query optimization. On the other hand, MQO can be used to search over such cases\cite{Roy2000,Zhou2007}.
\begin{figure}
\centering
\includegraphics[width=\linewidth,height=0.5\linewidth]{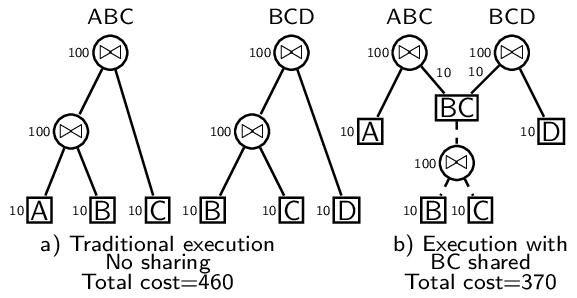}
\caption{MQO example (from \protect\cite{Roy2000}) illustrating benefit of sharing subexpressions}
\label{fig:mqoexample}
\end{figure}

While algorithms which find the optimal plan for a single query are well known, exhaustive algorithms for MQO take $\bigoh(n^n)$ time which quickly makes the problem untenable. Thus, work in this area relies on the development of various heurisics-based algorithms \cite{Zhou2007,Roy2000,Silva2012}. While most of such work seems to work well in practice, there has been no work which provides theoretical guarantees on the quality of solution obtained by any such heuristics, to the best of our knowledge. Thus, an open question is

\textit{Can we devise a polynomial-time algorithm which provides us with theoretical guarantees on the quality of the solution obtained as compared to the optimal? If so, what is the best possible polynomial-time approximation algorithm?}

As a first step towards answering this question, we propose a reformulation of the MQO problem, the motivation for which is stated next.

The canonical multi-query optimization problem is concerned with minimizing the cost of the query plan for a set of queries by choosing a set of nodes to materialize (say $M$) and then finding the optimal plan expoliting nodes in $M$. Another way to look at this problem is to maximize the ``materialization-benefit'' we get by materializing $M$ w.r.t. a naive execution plan which is locally optimal and does not exploit any common subexpressions. More formally, this corresponds to maximizing the difference of the cost of the best plan in which the set of materialized nodes is $M$ from the latter. As this is just a linear transformation of the cost function, it is clear that the maximizer of the materialization-benefit will be the minimizer of the cost.

Roy et al. \cite{Roy2000} assume a property which they call the ``monotonicity heuristic'' on the cost function. This essentially corresponds to assuming the supermodularity of the cost function defined on the set of nodes to be materialized. In \cite{Roy2000}, this assumption is used to speed up their greedy algorithm via a heap-based argument which exploits the supermodularity. This is similar to the \textsc{LazyGreedy} algorithm described in \cite{Minoux} for speeding up monotone submodular function maximization subject to cardinality constraints via the well-known greedy algorithm, which is also used by \cite{Roy2000}. On the queries used in their experiments, it was observed that the plan obtained with or without assuming supermodularity led to the same plan. This seems to imply that the supermodularity assumption may be a reasonable one and may hold in practice.

\subsection{Our contribution} 


The contributions of this paper are as follows
\begin{itemize}
  \item Motivated by \cite{Roy2000}, we proceed with the ``monotonicity heuristic'' assumption (which implies the submodularity of the materialization benefit function). \textit{Under this assumption}, we propose an approximation algorithm for the underlying problem of unconstrained, normalized submodular maximization (UNSM). Note that we allow the submodular function to take negative values, which has not been considered previously and poses a significant challenge.\footnote{Inapproximability results when the submodular function may be unnormalized are well known.}.
  \item We then present a hardness of approximation proof for the UNSM problem, which matches that obtained by our algorithm, under the weak assumption of $\p \neq \np$.  
  \item We present optimizations to our algorithm which can be used to improve the running time of the algorithm, without sacrificing any theoretical guarantees.
  \item We also consider a special case of the problem of submodular maximization under cardinality constraints. 
  \begin{itemize}
    \item A natural extension to our greedy algorithm for this problem is presented. We further propose a pruning strategy to reduce the search space before running our greedy algorithm, by exploiting this cardinality constraint.
    \item While, at this point, we do not formally prove any theoretical guarantees on the approximation factor for this constrained problem, we show that the answer obtained by our greedy algorithm is the same when run with or without this pruning.
  \end{itemize}
  \item We compare our algorithm against the Greedy algorithm and stand-alone Volcano (without MQO) on TPCD benchmark queries.
\end{itemize}
It is important to note that our approximation guarantees are for the benefit-maximization problem, under the submodularity assumption, and do not imply a multiplicative factor approximation to the cost minimization problem. However, results in our experimental section shows that our proposed algorithm performs as well as or better than the Greedy heuristic of \cite{Roy2000}.

Our techniques for the problem of multi-query optimization are presented in the context of query optimizers based on the Volcano/Cascades framework \cite{Graefe1993,Graefe1995}. This framework for optimizing queries uses transformation rules which makes it inherently extensible, and has been implemented in several widely-used commercial database systems such as Microsoft SQL Server. It should be noted, however, that our algorithm is agnostic to the query optimization framework and can be easily extended to other frameworks as well.


\textbf{Organization.} In Section \ref{Prelims}, we present a detailed overview of multi-query optimization in the context of the Volcano framework which was presented in \cite{Roy2000} along with how submodular maximization arises in this context. Section \ref{NewGreedy} presents our greedy algorithm for unconstrained, normalized submodular maximization with the proof of its approximation factor guarantee. In Section \ref{Hardness}, we prove the hardness of approximation of the unconstrained,normalized submodular maximization which rules out better approximation factors than the one attained by our algorithm, under the assumption of $\p \neq \np$. Section \ref{speedup} presents ways to speed up our algorithm. We present experimental results on benchmark queries in Section \ref{exptSection}. Related work in the areas of MQO and submodular maximization is presented in Section \ref{RelWork}. We conclude and discuss directions for future work in Section \ref{Conclusion}.

\vspace{-0.8em}
\section{Preliminaries}\label{Prelims}
This section presents some relevant background in (Multi)-Query Optimization in the Volcano framework followed by some preliminaries of submodular maximization and finally ends with how submodular maximization arises in MQO. Readers well-versed in MQO techniques in Volcano may skip to the third subsection directly.
\vspace{-0.5em}
\subsection{Query Optimization in Volcano}\label{VolcanoSQO}
\begin{figure}
\centering
\begin{subfigure}[b]{0.29\linewidth}
\centering
\includegraphics[width=1.5cm]{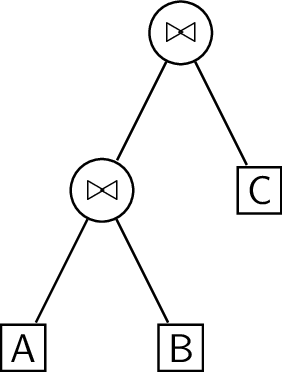}
\caption{Initial Query}
\label{fig:dagexa}
\end{subfigure}
\begin{subfigure}[b]{0.59\linewidth}
\centering
\includegraphics[width=1.5cm]{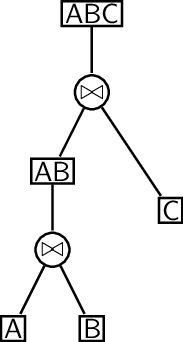}
\caption{DAG representation of query}
\label{fig:dagexb}
\end{subfigure}

\vspace{4mm}

\begin{subfigure}[b]{\linewidth}
\centering
\includegraphics[width=3.5cm]{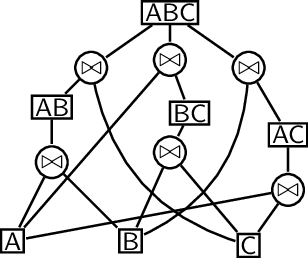}
\caption{Expanded LQDAG after transformation
(Commutativity not shown explicitly)}
\label{fig:dagexc}
\end{subfigure}
\caption{Initial Query and LQDAG Representation}
\label{fig:dagexpansion}
\vspace{-0.5em}
\end{figure}
The Volcano/Cascades query optimization framework \cite{Graefe1993,Graefe1995} is based on a system of equivalence rules, which specify that the result of a particular transformation of a query tree is the same as the result of the original query tree. The key aspect of this framework is the efficient implementation of the transformation rule-based approach.

The Volcano framework uses the AND-OR DAG representation \cite{Graefe1993,Roussopoulos:1982:VIR:319702.319729} for compactly representing the given query and its alternative query plans. An AND-OR DAG is a directed acyclic graph whose nodes can be divided into AND-nodes and OR-nodes; the AND-nodes have only OR-nodes as children and the OR-nodes have only AND-nodes as children. An AND-node corresponds to an algebraic operator, such as the join operator $(\bowtie)$ or a select operator $(\sigma)$. It represents the expression defined by the operator and its inputs. An OR-node represents a set of logical expressions that generate the same result set; the set of such expressions is defined by the children AND nodes of the OR node, and their inputs. Hereafter, we refer to the OR-nodes and AND-nodes as equivalence nodes and operator nodes respectively.

The given query tree is initially represented in the AND-OR DAG formulation. For example, the query tree of Figure \ref{fig:dagexa} is initially represented in the AND-OR DAG formulation, as shown in Figure \ref{fig:dagexb}. Equivalence nodes are shown as boxes, while operator nodes are shown in circles. 

The initial \AODAG is then expanded by applying all possible logical transformations on every node of the initial DAG created from the given query. Suppose the only possible transformations are join associativity and commutativity. Then the plans $A \bowtie (B \bowtie C)$ and $(A \bowtie C) \bowtie B$, as well as several plans equivalent to these, modulo commutativity, can be obtained by transformations on the initial \AODAG of Figure \ref{fig:dagexb}. These are represented in the DAG shown in Figure \ref{fig:dagexc}. The \AODAG representation after applying all the logical tranformations is called the (expanded) Logical Query DAG (or LQDAG).

Each operator node can have different physical implementations; for example, a join operator can be implemented as a hash join, a nested loop join or as a merge join. Once the LQDAG has been generated, physical implementation rules are applied on the logical operators to generate the physical AND-OR DAG, which is called the Physical Query DAG or PQDAG for short. 

Properties of the results of an expression, such as sort order, that do not form part of the logical data model are called physical properties \cite{Graefe1993}. The importance of exploiting physical properties such as sort order and partitioning of result sets is well known in traditional query optimization. The DAG is actually built and stored using a ``memo'' structure, a concise data structure used in the Volcano/Cascades framework to represent the entire space of equivalent query evaluation plans succintly. The AND-OR DAG representation considered for MQO actually works on the PQDAG but we present our algorithms to work at the LQDAG level for brevity. 

\subsection{Multi-Query Optimization in Volcano}
This subsection primarily focuses on the techniques presented in \cite{Roy2000} for MQO in the Volcano framework. In order to extend the Volcano \AODAG generation for MQO on a batch of queries to be jointly optimized, the queries are represented together in a single DAG, sharing subexpressions. The DAG is converted to a rooted DAG by adding a dummy operation node, which does nothing, but has the root equivalence nodes of all the queries as its inputs.

The two main challenges for a multi-query optimizer are : 
\begin{enumerate}
\item Recognizing possibilities of shared computation by identifying common subexpressions.
\item Finding a globally optimal evaluation plan exploiting the common subexpressions identified.
\end{enumerate}

Roy et al. \cite{Roy2000} present an efficient hashing-based algorithm that identifies the set of all common subexpressions, including subqueries that are syntactically different but semantically equivalent, in a single bottom-up traversal of the LQDAG by using the ``memo'' structure; for details see \cite{Roy2000}. This is similar to the ``expression fingerprinting'' used to identify the common subexpressions in \cite{Silva2012}. The combined LQDAG for the queries of Example \ref{firstExample} is shown in Figure \ref{fig:combinedDAG}. This step takes exponential time as the size of the DAGs may itself be exponential and is unavoidable, even in single-query optimization.

\begin{figure}
\centering
\includegraphics[width=0.75\linewidth,height=0.5\linewidth]{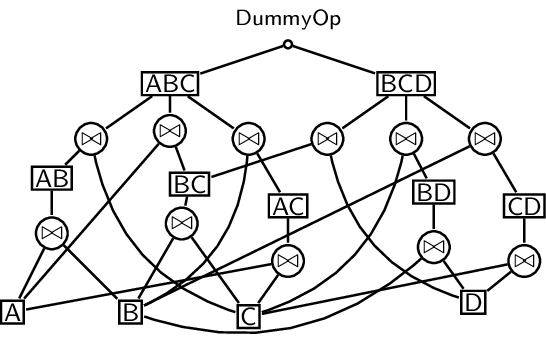}
\setlength{\abovecaptionskip}{-5pt}
\caption{Combined LQDAG for queries in Example \ref{firstExample}}
\label{fig:combinedDAG}
\end{figure}


 Similar to the single query optimization done by Volcano, in a single-pass, one can annotate each node in the DAG with its estimated cost. Note that the cost estimator functions are taken as input to the optimizer, i.e., the optimizer algorithm is agnostic to the cost estimates. Indeed, this is one of the reasons why the Volcano query optimizer framework is widely used. It is important to note that in the single query optimization as well as the multi query optimization setting, one assumes that the cost estimates provided to us are correct for any guarantees to hold. Thus, we also work under the assumption that the cost estimates are correct. After the common subexpressions are identified and the cost of each node computed, the next task is to find the best consolidated plan for the queries exploiting the subexpressions. 

 In this paper, we are primarily concerned with the optimization philosophy adopted by the Greedy algorithm in \cite{Roy2000} which is presented next. For a set of equivalence nodes $S$, let \textit{bestCost}$(Q,S)$ (for brevity, $bc(S)$) denote the cost of the optimal plan for Q given that nodes in S are to be materialized (this includes the cost of computing and materializing nodes in S). Here $Q$ is the combined query DAG with the dummy root operator node with inputs being the DAGs of $Q_1, \ldots, Q_k$, as described above. The $bc(S)$ function, of course, depends on the cost estimates and is treated as a black-box for the MQO algorithms. Given a set of nodes $S$ to be materialized, \cite{Roy2000} present an efficient scheme to find the best plan and the best cost, $bc(S)$ (this includes the cost of materializing $S$, which may be done in multiple ways and is figured out by the optimizer in \cite{Roy2000} as well).

 Now, we just need to identify the subset $S$ of nodes in the \AODAG for which \textit{bestCost}$(Q,S)$ is minimum. However, an exhaustive algorithm which enumerates all possible subsets $S$ will take time exponential in the size of the \AODAG, which itself may be exponential in size. In \cite{Roy2000}, they propose an intuitive greedy algorithm which iteratively picks which node to materialize. At each iteration, the node $x$ that gives the maximum reduction in the cost, if materialized, is chosen to be added to the current set of materialized nodes $X$. While this greedy algorithm is shown to work well in practice, they \cite{Roy2000} do not theoretically argue about the quality of solution obtained via this algorithm. The algorithm is presented below for completeness.

\begin{algorithm}
\caption{Greedy Algorithm of \protect\cite{Roy2000}}
\begin{algorithmic}[ht]
\State $X = \emptyset$
\State $Y =$ Set of shareable equivalence nodes in the DAG 
\While {$Y \neq \emptyset $}
\State Pick $x \in Y$ which minimizes $bc(X \cup \{x\})$
\If{$bc(X) > bc(X\cup\{x\})$}
\State $X=X\cup\{x\}, Y = Y \setminus \{x\}$
\Else
\State $Y=\emptyset$
\EndIf 
\EndWhile\\
\Return $X$
\end{algorithmic}
\end{algorithm}
\vspace{-0.4em}
As noted in \cite{Roy2000}, the nodes materialized in the globally optimal plan are just a subset of the ones that are shared in some plan for the query. It is, thus, sufficient to search only over the set of \textit{shareable} equivalence nodes, instead of searching over the entire set of equivalence nodes in the DAG.

Clearly, some assumptions on the cost function have to be made in order to give theoretical guarantees for any algorithm. Furthermore, it is desirable to make assumptions which may hold in practice. Roy et al. \cite{Roy2000} make an additional assumption which they call the ``monotonicity heuristic''. \\
Define $benefit(x,X)$ as $bc(X) - bc(X \cup \{x\})$. The assumption is that 
\begin{align*}\forall \ Y \subseteq X, \  \forall \ x \notin X, \ benefit(x,X) \leq benefit(x,Y). \end{align*}
They \cite{Roy2000} make this assumption in order to improve the running time of their greedy algorithm via a heap-based argument which corresponds to the \textsc{LazyGreedy} algorithm \cite{Minoux} for faster monotone, submodular maximization. Their experiments, however, show that the plans obtained with and without the assumption had exactly the same cost. While the assumption may not always hold, their experiments seem to indicate that the assumption may be a reasonable one, in practice. Thus, in this paper, we work under this assumption to devise an algorithm \textit{with theoretical guarantees} on its performance for maximizing the ``materialization benefit''.

\subsection{Submodular Maximization}
Let $U$ be a universe of $n = |U|$ elements, let $f : 2^U \rightarrow \R$ be a function. For simplicity, we use the notation $f'(u,S)$ to denote the incremental value in $f$ of adding $u$ to $S$, i.e., $f'(u,S) = f(S\cup\{u\}) - f(S)$.

\begin{definition}
(\small{\textsc{Submodular Functions}})\\ A function $f : 2^U \rightarrow \R$ is called \textit{submodular} if \begin{align*}\forall \ A \subseteq B \subseteq U, \ \forall \ u \in U \setminus B, \mbox{we have } f'(u,A) \geq f'(u,B).
\end{align*}
\end{definition}
 \vspace{-1em}

\begin{definition}
(\small{\textsc{Supermodular Functions}})\\ A function $f : 2^U \rightarrow \R$ is called \textit{supermodular} if \begin{align*}\forall \ A \subseteq B \subseteq U, \ \forall \ u \in U \setminus B, \mbox{we have } f'(u,A)\leq f'(u,B).
\end{align*}
\end{definition}
 \vspace{-1em}

\begin{definition}
(\small{\textsc{Additive Functions}})\\ A function $c : 2^U \rightarrow \R$ is called \textit{additive} if it is of the form $c(S) = \sum_{e \in S} c(\{e\})$. 
\end{definition}
 \vspace{-1em}
\begin{definition}(\small{\textsc{Monotone Functions}})\\
A function $f : 2^U \rightarrow \R$ is said to be \textit{monotone} if \begin{align*} \forall A \subseteq B \subseteq U, \mbox{we have } f(A) \leq f(B).\end{align*}
\end{definition}
 \vspace{-1em}

\begin{definition}
(\small{\textsc{Normalized Functions}})\\ A function $f : 2^U \rightarrow \R$ is called \textit{normalized} if $f(\emptyset) = 0.$

\end{definition}

Given a normalized submodular function $f : 2^U \rightarrow \R$, the unconstrained, normalized submodular maximization (UNSM) problem is to find a set $S \subseteq U$ which maximizes the value of $f$, i.e., $\arg\max\limits_{S \subseteq U} f(S)$.

Since submodular maximization problems are in general $\mathsf{NP}$-hard and can only be approximated, a simple additive scaling of the function by a large constant to make the function non-negative and running an algorithm like \cite{Buchbinder2012} suffers in the approximation factor and moreover does not guarantee a multiplicative approximation.

It is well-known that any non-monotone submodular function $f$, with the constraint that $f(\emptyset) = 0$, can be written as the difference of a non-negative monotone submodular function $f_M$ and an additive ``cost'' function $c$. However, multiple such decompositions are possible and as we will show, there is one particular decomposition (the decomposition in Proposition \ref{decomp}) which will give us the best approximation ratio and a matching hardness of approximation.
\begin{proposition}\label{decomp} Any normalized, non-monotone (which may take negative values) submodular function $f$ can be decomposed as \begin{align*} f(S) = f_M(S) - c(S) \ \ , \forall \ S \subseteq U\end{align*} where $f_M$ is a monotone submodular function and $c$ is an additive cost function. In particular, one possible decomposition is\small
\begin{align*}
& f_M^*(S) = f(S) + \sum\limits_{e \in S}(f(U \setminus \{e\})-f(U))\\
& c^*(S) = \sum\limits_{e \in S}(f(U \setminus \{e\}) - f(U)) 
\end{align*}
\end{proposition}\normalsize
\begin{proof} The proof is provided in Appendix A. 
\end{proof}
Since our approximation ratio depends on the decomposition and owing to the importance of the decomposition in Proposition \ref{decomp}, we refer to it as $f_M^*$ and $c^*$.

\subsection{Multi-Query Optimization and UNSM}\label{mqounsm}
We now describe the changes to the MQO formulation of \cite{Roy2000} and show the role submodularity plays in the same. As defined above, \textit{bestCost}$(Q,S)$ includes the cost of computing and materializing the set of PQDAG nodes to be materialized $S$. Consider a scenario where $S$ was already materialized and we just have to find the optimal plan which \textit{may or may not} use the materialized nodes in $S$. However, no further nodes may be chosen to be materialized. The cost of the optimal plan can be thought of as the \textit{best use cost} and the function is thus called \textit{bestUseCost}$(Q,S)$. This function is monotonically decreasing since as more nodes are materialized, we will exploit the additional nodes only if they lead to a reduction in cost. Of course, the cost of materializing $S$ needs to be taken into account and we call that function $c(S)$. Clearly, \textit{bestCost}$(Q,S) = $ \textit{bestUseCost}$(Q,S)+c(S)$. For brevity, we refer to \textit{bestUseCost}$(Q,S)$ as $buc(S)$. 

The MQO problem can be thought of as maximizing the ``materialization-benefit'' $(mb(S) $ for brevity) we get in the plan cost by exploiting common subexpressions over a naive execution plan which is just locally optimal and does not exploit subexpressions.  Clearly the cost of the latter is $bc(\emptyset) = buc(\emptyset)$. Mathematically, $mb(S)$ is defined as
\begin{align*}
mb(S) &= bc(\emptyset) - bc(S) \\
&= buc(\emptyset) - (buc(S) + c(S)) \\
&= (buc(\emptyset) - buc(S)) - c(S) 
\end{align*}
The function in parenthesis in the last line is a monotonically increasing function since $buc(S)$ is a monotonically decreasing function. Also, if the set of materialized nodes $S$ are ``far apart'' in the PQDAG, the cost of computing and materializing a node $e \in S$ can be thought of as being independent of the other nodes in $S$. This motivates us to assume that the $c$ function is additive. Of course, this assumption need not be true. For example, if two of the equivalence nodes in $S$ are just below each other, we can significantly benefit by computing the ``lower'' node and then just reading it from disk to compute the ``upper'' node. As proved in Proposition \ref{decomp}, under the assumption of submodularity, $mb$ can always be decomposed into a difference of monotone, submodular function and an additive function\footnote{The decomposition in Proposition \ref{decomp} does not actually correspond to the cost of materializing nodes but parallels are drawn for intuition}. Observe that \begin{align*}\forall X,\ \forall x \notin X, benefit(x,X) = -bc'(x,X)\end{align*} Thus, the ``monotonicity heuristic'' assumption is essentially that the \textit{bestCost} function is supermodular. This implies that $mb$ is submodular. Note that $mb$ is normalized. Thus, the problem is essentially the UNSM problem with $mb$ as the submodular function.
The reason why materialization benefit for a particular set of nodes may be negative is due to the fact that there may be certain nodes which may have very high materialization cost but may not have high benefit. This is where the algorithm of \cite{Silva2012}, which chooses to materialize every node can be horribly inefficient.
\vspace{-0.8em}
\section{The Marginal Greedy algorithm}\label{NewGreedy}
 In this section, we propose a greedy algorithm for the UNSM problem for which we prove an approximation guarantee in this section. A proof of a matching hardness of approximation, under the assumption of $\p \neq \np$  is presented in the next section.

Given a decomposition of a non-monotone, normalized submodular function $f$, let the monotone submodular and additive functions be denoted by $f_M$ and $c$. Thus, the problem we want to solve is as follows
\begin{align*}
 & \max_{\substack{S \subseteq U}} f(S)
 = \max_{\substack{S \subseteq U}} f_M(S) - c(S)
\end{align*}
The \margG algorithm (Algorithm \ref{margGreedy}) has been proposed before by \cite{DBLP:journals/orl/Sviridenko04}, albeit for non-negative, \textit{monotone} submodular maximization under knapsack constraints. At each iteration, the algorithm greedily selects the element with the highest use-benefit to cost ratio from those elements which satisfy a knapsack constraint. In our case, however, there is no knapsack constraint and instead we add elements as long as it leads to an increase in the value of $f$. We emphasize that the problem in our case is considerably different than this problem and highlight the differences in subsection \ref{approxsubsection}.

\begin{algorithm}
\captionof{algorithm}{\margG Algorithm}\label{margGreedy}
\begin{algorithmic}[ht]
\State $X = \emptyset$
\State $Y =$ Set of shareable equivalence nodes in the DAG 
\While {$Y \neq \emptyset $}
\State Pick $x \in Y$ which maximizes $r(x, X) = \frac{f_M'(x,X)}{c(\{x\})}$
\If{$r(x,X) > 1$}
\State $X=X\cup\{x\}, Y = Y \setminus \{x\}$
\Else
\State $Y=\emptyset$
\EndIf 
\EndWhile\\
\Return $X$
\end{algorithmic}
\end{algorithm}
The \margG algorithm also finally adds all elements with negative $c$ values. This was also done in Sviridenko's case \cite{DBLP:journals/orl/Sviridenko04} as one can only increase the value of the function without increasing the budget. This is fine for us as well and can only raise the value of the function $f$. This is because $f_M$ is monotone so including more elements only raises its value and we are subtracting off some negative $c$ values which can only raise the value of $f$. If the decomposition used is the one given in Proposition \ref{decomp}, we can compute the term in the summation for each element once and store it. This can be done in just $n+1$ $bc(S)$ invocations (for the sets $U$ and for $U \setminus \{e_i\} \ \forall e_i \in V$).

\subsection{Approximation Factor of \margG}\label{approxsubsection}

Let $\Theta$ be an optimal solution. Let $X_i$ denote the set of nodes selected by Algorithm \ref{margGreedy}  just after the $i^{th}$ iteration. \\
Define $\Delta_{f_M}(E,S) = f_M(S\cup E) - f_M(S)$, where $E$ and $S$ are subsets of $U$.

We state the main theorem of this section which mentions the approximation guarantee Algorithm \ref{margGreedy} provides. The approximation factor is not a constant and instead depends on the value of the $f$ and $c$ functions at optimal. 
\begin{theorem}\label{approxFactor}
The answer obtained by the \margG algorithm ($X$) satisfies the following inequality\small
\begin{align*}
f(X) \geq \bigg[1-\frac{c(\Theta)}{f(\Theta)}\ln (1+\frac{f(\Theta)}{c(\Theta)} )\bigg] f(\Theta).
\end{align*}
\end{theorem}\normalsize
We prove the theorem after presenting a lemma and its corollary which are central to the proof. At a high level, the lemma states that upto a certain point in the execution of the  algorithm, there exists an element that can be picked and has a marginal-benefit to cost ratio which is at least the marginal-benefit to cost ratio we would get if we picked all remaining elements in the optimal solution.
\begin{lemma}\label{structuralLemma}
At any iteration $i + 1 < n$ in the execution of the \margG algorithm, if $f_M(X_i) < f(\Theta)$, then there exists some element $e \in \Theta \setminus X_i$ that satisfies
\begin{align*}
\frac{\Delta_{f_M}(\{e\},X_i)}{c(\{e\})} \geq \frac{\Delta_{f_M}(\Theta,X_i)}{c(\Theta)}.
\end{align*}
\end{lemma}
\begin{proof} 
Firstly, note that if \\
$f_M(X_i) < f(\Theta) = f_M(\Theta) - c(\Theta) \leq f_M(\Theta)$, then \\ $ \Theta \setminus X_i \neq \emptyset$. This is because $f_M$ is monotonically increasing. Also, note that if $S$ is fixed, $\Delta_{f_M}(E,S)$ is a submodular function in $E$, due to submodularity of $f_M$. 

We consider two cases. Since the $f_M$ function is monotonically increasing, the numerators on both sides of the inequality are non-negative.

\textbf{Case 1.} $\Delta_{f_M}(\Theta,X_i) = 0$\\
In this case, the RHS of the inequality is 0. Since the $f_M$ function is monotonically increasing, $\forall e' \in \Theta \setminus X_i$, we have $\frac{\Delta_{f_M}(e',X_i)}{c(\{e'\})} \geq \frac{\Delta_{f_M}(\Theta,X_i)}{c(\Theta)}$. Since $\Theta \setminus X_i \neq \emptyset$, any element $e' \in \Theta \setminus X_i$ satisfies the required inequality.

\textbf{Case 2.} $\Delta_{f_M}(\Theta,X_i) > 0$\\
We first show that there exists some element $e \in \Theta$ for which the inequality holds.
Assume the contradiction, i.e., 
\begin{align*} 
\forall e \in \Theta, & \frac{\Delta_{f_M}(\{e\},X_i)}{c(e)} < \frac{\Delta_{f_M}(\Theta,X_i)}{c(\Theta)}.\end{align*}\vspace{-1em}
\begin{align*}
  \therefore c(e)(\Delta_{f_M}(\Theta,X_i) ) > c(\Theta)(\Delta_{f_M}(\{e\},X_i)).
 \end{align*}

  Summing up over all $ e \in \Theta$, we get
  \small
  \begin{align*}& \sum\limits_{e \in \Theta}c(e)(\Delta_{f_M}(\Theta,X_i)) > \sum\limits_{e \in \Theta}c(\Theta)(\Delta_{f_M}(\{e\},X_i) )\\
  \implies& (\Delta_{f_M}(\Theta,X_i) ) \sum\limits_{e \in \Theta}c(e) > c(\Theta)\sum\limits_{e \in \Theta}(\Delta_{f_M}(\{e\},X_i) )\\
  \implies& (\Delta_{f_M}(\Theta,X_i) ) c(\Theta) > c(\Theta)\sum\limits_{e \in \Theta}(\Delta_{f_M}(\{e\},X_i) )\\
  \implies& \Delta_{f_M}(\Theta,X_i)  > \sum\limits_{e \in \Theta}(\Delta_{f_M}(\{e\},X_i)).
  \end{align*}
  \normalsize
  Since $X_i$ is fixed, from our earlier observation, $\Delta_{f_M}(E,X_i)$ is a submodular function in $E$. 
  Thus, we have
\begin{align*}
  \Delta_{f_M}(\Theta,X_i)  \leq \sum\limits_{e \in \Theta}(\Delta_{f_M}(\{e\},X_i)).
\end{align*}
This leads to a contradiction. Thus, there exists some element $e' \in \Theta$ for which the required inequality holds. 

Now, observe that the RHS of the required inequality in this case is strictly positive and $\forall e \in X_i$, the LHS of the inequality is 0. Hence, $e' \notin X_i$ and we are done.
\end{proof}
\vspace{-0.5em}
\begin{corollary}\label{structuralCorollary}
When the conditions of Lemma \ref{structuralLemma} hold, 
\begin{center}
$\frac{\Delta_{f_M}(\{e\},X_i) - c(\{e\})}{\Delta_{f_M}(\{e\},X_i)} \geq \frac{\Delta_{f_M}(\Theta,X_i) - c(\Theta)}{\Delta_{f_M}(\Theta,X_i)}.$
\end{center}

\end{corollary}
\begin{proof}
From Lemma \ref{structuralLemma}, we have
\begin{center}
$ \frac{\Delta_{f_M}(\{e\},X_i)}{c(\{e\})} \geq \frac{\Delta_{f_M}(\Theta,X_i)}{c(\Theta)}.$
\end{center}
Since $f_M$ is monotonically increasing, it implies
\begin{center}
$\frac{\Delta_{f_M}(\{e\},X_i) - c(\{e\})}{\Delta_{f_M}(\{e\},X_i)}  \geq \frac{\Delta_{f_M}(\Theta,X_i) - c(\Theta)}{\Delta_{f_M}(\Theta,X_i)},$
\end{center}
and we are done.
\end{proof}

\begin{proof} (of Theorem \ref{approxFactor})
Say the \margG algorithm runs for $l \leq n$ iterations. Define $\alpha(X_i)$ to be the rate of increase of $f$ with respect to $f_M$ just after the $i^{th}$ iteration (and thus the current chosen set of elements is $X_i$). Further, let $e \in U \setminus X_i$ be the next element that will be chosen by the \margG algorithm. Note that $e$ is actually a function of $X_i$ and, thus, once $X_i$ is fixed, so is $e$. Mathematically, 
\small
\begin{displaymath}
\alpha(X_i) = \frac{f(X_i\cup\{e\}) - f(X_i)}{\delta(f_M(X_i))}
\end{displaymath}
where $\delta(f_M(X_i)) = f_M(X_i\cup\{e\}) - f_M(X_i).$
 
 Let $j \leq l$ be the maximal index such that $f_M(X_j) < f(\Theta)$. The rate of increase at iteration $i$ of the algorithm is at least as large as choosing the element from $\Theta \setminus X_i$ with the rate presented in LHS of Corollary \ref{structuralCorollary}. 

The corollary also implies that while $f_M(X_i) < f(\Theta)$, the greedy algorithm has an element that it can pick. This implies that $j < l$. Thus, we have
\begin{center}
$f(X_l) = \sum_{i=0}^{l-1}\alpha(X_i)\delta(f_M(X_i)).$
\end{center}

Using Corollary \ref{structuralCorollary},
\small
\begin{align*}
f(X_l)&\geq  \sum_{i=0}^{l-1}\bigg(\frac{f_M(\Theta) - f_M(X_i)-c(\Theta)}{f_M(\Theta)-f_M(X_i)}\bigg)\delta(f_M(X_i)) \\
 &\geq  \sum_{i=0}^{l-1}\bigg(1 - \frac{c(\Theta)}{f_M(\Theta)-f_M(X_i)}\bigg)\delta(f_M(X_i)).
\end{align*}
\normalsize

Since the term in the parenthesis in the last line is a decreasing function of $f_M(X_i)$, we get
\small
\begin{align*}
f(X_l) \geq & \int\limits_{0}^{f_M(X_l)}\bigg(1 - \frac{c(\Theta)}{f_M(\Theta)-u}\bigg)du\\
\geq &  \int\limits_{0}^{f(\Theta)}\bigg(1 - \frac{c(\Theta)}{f_M(\Theta)-u}\bigg)du\\
= & \bigg[u + c(\Theta)\ln(f_M(\Theta) - u) \bigg]^{f(\Theta)}_{0}\\
= & f(\Theta) + c(\Theta)\ln\bigg(\frac{f_M(\Theta)-f(\Theta)}{f_M(\Theta)}\bigg)\\
= & f(\Theta) + c(\Theta)\ln\bigg(\frac{c(\Theta)}{f(\Theta)+c(\Theta)}\bigg)\\
= & f(\Theta) - c(\Theta)\ln\bigg(\frac{c(\Theta)+f(\Theta)}{c(\Theta)})\bigg)\\
= & f(\Theta) - c(\Theta)\ln\bigg(1+\frac{f(\Theta)}{c(\Theta)})\bigg)\\
= & \bigg[1 - \frac{c(\Theta)}{f(\Theta)}\ln\bigg(1+\frac{f(\Theta)}{c(\Theta)}\bigg) \bigg]f(\Theta).
\end{align*}

This concludes our proof and gives us our required approximation factor of 
\small$
\bigg[1 - \frac{c(\Theta)}{f(\Theta)}\ln\bigg(1+\frac{f(\Theta)}{c(\Theta)}\bigg)\bigg].
$\end{proof}
\vspace{-0.4em}
Since the approximation ratio depends on the decomposition (specifically the function $c$), it is natural to ask whether different decompositions can lead to different solutions and approximation ratios. This is indeed the case; given a decomposition $f_M$ and $c$, we can add a positive linear function $d(S) = \sum_{i \in S} d_i$ to both $f_M$ and $c$, we still have a valid decomposition and the approximation factor has become smaller. This is because $f(\Theta)$ is fixed but $c(\Theta)$ becomes larger and clearly, the ratio is a decreasing function of $c$. Since this is the case, one may ask what is the ``best'' decomposition for this problem? We now show that the decomposition in Proposition \ref{decomp}, $f_M^*$ and $c^*$, is indeed the best decomposition. This is done by first improving the ratio for an arbitrary decomposition and then showing that the improvement procedure for $f_M^*$ and $c^*$ does not lead to any improvement. In fact, in the next section, we will show a hardness of approximation which matches the ratio provided by this decomposition.

First we show how to obtain from an arbitrary decomposition $f_M$ and $c$, another decomposition $\widetilde{f}_M$ and $\widetilde{c}$ such that the ratio improves. This happens if we can subtract a linear term from $f_M$ and $c$ while preserving monotonicity of $f_M$ based on the above argument. And then we show that for $f_M^*$ and $c^*$, this improvement procedure returns $f_M^*$ and $c^*$
\begin{proposition} Given an arbitrary decomposition $f_M$ and $c$ of a normalized submodular function $f$, i.e., $f(S) = f_M(S) - c(S) \ \forall \ S \subseteq V$ with monotone $f_M$ and consider another decomposition 
\small
\begin{align*}
\widetilde{f}_M (S) &= f_M (S) - \sum_{i \in S} \big(f_M (U) - f_M (U \setminus i) \big) \\
\widetilde{c} (S) &= c(S) - \sum_{i \in S} \big(f_M (U) - f_M (U \setminus i) \big)
\end{align*}
\normalsize
Then, $\widetilde{f}_M$ is monotone. Furthermore, for the decomposition in Proposition \ref{decomp}, $f^*_M$ and $c^*$, $\widetilde{f}^*_M = f^*_M$ and $\widetilde{c}^* = c^*$. 
\end{proposition}

\begin{proof} The proof is provided in Appendix A.
\end{proof}

We now remark on certain aspects of the algorithm and its analysis. Since the algorithm is inspired by \cite{DBLP:journals/orl/Sviridenko04}, one may ask whether running that algorithm for multiple values of the budget in the knapsack constraint leads to the same answer. Indeed, this is the case with budget being the value of $c(\Theta)$. However, since we do not apriori know $c(\Theta)$, we would have to potentially try out a large number of budget values which is not feasible. Furthermore, our analysis of the approximation ratio crucially uses the fact that we are actually running the algorithm on this decomposition of $f$ in order to maximize $f$ itself and not maximizing a monotone submodular function subject to knapsack constraints.

\vspace{-1em}
\section{Inapproximability of UNSM}\label{Hardness}
In this section, we prove a hardness of approximation result for the UNSM problem which matches the approximation factor given by the \margG algorithm in Theorem \ref{approxFactor} when the decomposition used is $f_M^*$ and $c^*$ as defined in Proposition \ref{decomp}.
\begin{theorem}\label{hardnessThm} For any $\varepsilon > 0$, it is NP-hard to approximate the unconstrained, normalized submodular maximization problem to a factor of at least \small\begin{align*}\bigg(1 - \frac{\ln(1+\gamma)}{\gamma}+\varepsilon\bigg).\end{align*}\normalsize Here, $\gamma = \frac{f(\Theta)}{c^*(\Theta)}$ and $\Theta$ is an optimal solution to the unconstrained, normalized submodular maximization problem.
\end{theorem}
This approximation factor depends on the value at optimal (which may go to 0), implying that a constant factor approximation to the UNSM problem is unlikely. 

Before proving Theorem \ref{hardnessThm}, we first present a separation result of the \textsf{Max Coverage} problem which is central to the proof of Theorem \ref{hardnessThm}.
\subsection{Inapproximabili{}ty of \textsf{Max Coverage}}
An instance $\mathcal{I} = (X, \mathcal{S})$ of the \textsf{Set Cover} problem is defined as follows: we are given the ground set $X = \{e_1, e_2, \ldots, e_n\}$ and $\mathcal{S} = \{S_1, S_2, \ldots, S_m\} \subseteq 2^X$. The goal is to choose the minimum number of sets $\mathcal{O} \subseteq \mathcal{S}$ such that $\bigcup \limits_{S_i \in \mathcal{O}}S_i = X$. 
Feige \cite{Feige1998} showed that for any $\varepsilon > 0$, there is no $(1 - \varepsilon) \ln n$-approximation polynomial time algorithm for this problem unless $\np \subseteq \mathsf{DTIME}(n^{\textit{O}(\log\log n)})$. The hardness was later proved under the weaker assumption of $\p \neq \np$ by \cite{Moshkovitz2015,Dinur:2014:AAP:2591796.2591884}.

A problem closely related to the \textsf{Set Cover} problem is the \textsf{Max Coverage} problem. An instance of the \textsf{Max Coverage} problem consists of an instance $\mathcal{I} = (X, \mathcal{S}, l)$ where $X$ is the ground set, $\mathcal{S}$ is a collection of subsets of $X$, and $l \leq m$ is an integer specifying the budget. The goal is to select $l$ sets $S_{i_1}, S_{i_2} \ldots, S_{i_l}$ and cover as many elements of the ground set as possible. Feige \cite{Feige1998} shows that it is $\np$-hard to approximate this problem to  a factor better than $1 - 1/e$. 

Krishnaswamy and Sviridenko \cite{Krishnaswamy2012} prove a separation result (which is an extension of the \textsf{Max Coverage} hardness stated above) which is of interest to us. 
 \vspace{-1mm}
\begin{theorem}\label{betaSep} (Theorem 2.2 in \cite{Krishnaswamy2012}) Suppose there exists a polynomial algorithm, which for some constants $B \geq 1$ and $0 < \varepsilon < e^{-B}$  has the following property : Given any instance $(X,\mathcal{S}, l)$ of \textsf{Max Coverage} with optimal value equal to $|X|$ (i.e., there exist $l$ sets that cover the ground set $X$ completely), the algorithm picks a collection of $\beta l$ sets for some $\beta \in [0,B]$ which can cover $(1 - e^{-\beta}+\varepsilon)n$ elements. Then $\p = \np$. Note that we allow the algorithm to pick different values of $\beta$ for different instances of the problem.
\end{theorem}
 \vspace{-1mm}

Theorem 2.2 in \cite{Krishnaswamy2012} is actually stated under the stronger assumption of $\np \not\subseteq \mathsf{DTIME}(n^{\textit{O}(\log\log n)})$. Their reduction relies on the hardness of \textsf{Set Cover} which, at the time of that paper, was known only under this stronger assumption. Leveraging the set cover hardness result by \cite{Moshkovitz2015,Dinur:2014:AAP:2591796.2591884} under the weaker assumption of $\p \neq \np$, we arrive at Theorem \ref{betaSep} without any changes to the proof provided in \cite{Krishnaswamy2012}.

Note that the coverage function $f(\mathcal{A}) = \bigg|\bigcup\limits_{S \in \mathcal{A}}S\bigg|$ is a monotone, submodular function. The proof of Theorem \ref{hardnessThm} proceeds by considering a special case of UNSM where for a \textsf{Max Coverage} instance, $f_M(\mathcal{A})$ is taken to be a scaling of the coverage function and the additive cost function $c(\mathcal{A})$ is a scaling of the cardinality of the chosen set of subsets $\mathcal{A}$. We call this the \textsf{Profitted Max Coverage} problem. 
\begin{problem}(The \textsf{Profitted Max Coverage} problem)\label{profitMCProb}
An instance of this problem consists of an instance $\mathcal{I}=(X, \mathcal{S}, l)$ like the \textsf{Max Coverage} problem. Consider $\gamma$ to be a constant for this problem whose value will be revealed later. 

Let $f_M(\mathcal{A}) = \frac{(\gamma + 1)}{\gamma}\frac{\big|\bigcup\limits_{S \in \mathcal{A}}S\big|}{n}$ and $c(\mathcal{A}) = \frac{1}{\gamma}\frac{|\mathcal{A}|}{l}$. The goal is to maximize 
\small
\begin{align*}
f(\mathcal{A}) & = f_M(\mathcal{A}) - c(\mathcal{A}) \\
& = \frac{(\gamma + 1)}{\gamma}\frac{\bigg|\bigcup\limits_{S \in \mathcal{A}}S\bigg|}{n} - \frac{1}{\gamma}\frac{\big|\mathcal{A}\big|}{l}
\end{align*}
\end{problem}
\normalsize
\begin{proof}(of Theorem \ref{hardnessThm}) 
We want to show that if there exists a polynomial time algorithm which approximates the \textsf{Profitted Max Coverage} problem to a ratio better than \small\begin{align*}1 - \frac{\ln(\gamma+1)}{\gamma}+\varepsilon\frac{(\gamma + 1)}{\gamma},\end{align*}\normalsize then $\p = \np$.

We consider a hard instance $\mathcal{I} = (X, \mathcal{S},l)$ of the \textsf{Max Coverage} problem such that the optimal value is $n$ (i.e., there exist $l$ sets to cover the entire ground set $X$). Now, let functions $f, f_M \mbox{ and } c$ be defined as in Problem \ref{profitMCProb}.

[\textit{Completeness}] Let us take a collection of $l$ sets $\mathcal{G} = \{S_{i_1}, S_{i_2}, \ldots, S_{i_l} \}$ that cover the ground set X (such a collection exists because $\mathcal{I}$ is a \textsf{Max Coverage} instance with optimal value $n$). The optimal value of the corresponding \textsf{Profitted Max Coverage} instance occurs when exactly the sets in $\mathcal{G}$ are chosen. 
\small
\begin{align*}
f(\mathcal{G}) &= \frac{(\gamma + 1)}{\gamma}\frac{n}{n} - \frac{1}{\gamma}\frac{l}{l} \\
&= \frac{(\gamma + 1)}{\gamma} - \frac{1}{\gamma}\\
&= 1.
\end{align*}
\normalsize
Observe that $\frac{f(\mathcal{G})}{c(\mathcal{G})}=\gamma$.

[\textit{Soundness}] It is easy to see that we will never choose more than $(\gamma + 1)l$ sets as the function $f$ will take negative values in those cases. 

For any set, say $\mathcal{F}$, of $\beta l$ (where $\beta \in [0,\gamma + 1]$) subsets from $\mathcal{S}$  which cover at most $(1 - e^{-\beta}+\varepsilon)n$ elements, the value of the \textsf{Profitted Max Coverage} instance in this case is at most:
\small
\begin{align*}
f(\mathcal{F}) &\leq \frac{(\gamma + 1)}{\gamma}\frac{(1-e^{-\beta}+\varepsilon)n}{n} - \frac{1}{\gamma}\frac{\beta l}{l} \\
&= \frac{(\gamma + 1)}{\gamma}(1-e^{-\beta}+\varepsilon) - \frac{1}{\gamma}\beta\\
&= \frac{(\gamma + 1)(1-e^{-\beta}+\varepsilon) - \beta}{\gamma}.
\end{align*}
\normalsize
Differentiating the expression in the last line w.r.t $\beta$ and setting the derivative to 0, we get
\begin{align*}
 &\frac{\gamma+1}{\gamma}(e^{-\beta}) - \frac{1}{\gamma} = 0\\
 \implies&e^{\beta} = (\gamma + 1)\\
\implies & \beta = \ln (\gamma + 1) \leq (\gamma + 1).
\end{align*}

Thus, the value $f(\mathcal{F})$ is always less than the value attained for that value of $\beta$ and is 
\begin{align*}f(\mathcal{F})&\leq 1 - \frac{\ln(\gamma+1)}{\gamma}+\varepsilon\frac{(\gamma + 1)}{\gamma}.\end{align*}

Now, if there exists a polynomial time algorithm (say \textsf{Alg}) which solves  the \textsf{Profitted Max Coverage} problem to a factor better than $1 - \frac{\ln(\gamma+1)}{\gamma}+\varepsilon\frac{(\gamma + 1)}{\gamma}$, then on any input instance of the \textsf{Max Coverage} problem such that the optimal value is $n$, \textsf{Alg} will output a set $\mathcal{F}$ such that $f(\mathcal{F}) > 1 - \frac{\ln(\gamma+1)}{\gamma}+\varepsilon\frac{(\gamma + 1)}{\gamma}$ (since the optimal value is 1). Thus, $\mathcal{F}$ covers strictly more than $(1 - e^{-\beta}+\varepsilon)n$ elements with $\beta = \frac{|\mathcal{F}|}{l}$ (by contrapositivity). By Theorem \ref{betaSep}, we have $\p = \np$.

The above argument establishes the hardness for $\gamma = \frac{f(\Theta)}{c(\Theta)}$ for the function $c$ defined in Problem \ref{profitMCProb}. Since the factor depends only on $c(\Theta)$, if we can show that $c(\Theta) = c^* (\Theta)$ for these hard instances, we would be done. This can be shown by considering the expression for $c^*(\Theta)$ in this case : 
\small
\begin{align*}
c^*(\Theta) &= \sum_{i \in \Theta} \big(f(U\setminus\{i\}) - f(U)\big)\\
&=\sum_{i \in \Theta} \big(f_M (U \setminus \{i\}) - f_M(U) - c(U\setminus\{i\}) + c(U)\big)\\
&= \sum_{i \in \Theta} \big(c(U) - c(U\setminus\{i\})\big) + \sum_{i \in \Theta} \big(f_M (U \setminus \{i\}) - f_M(U)\big) \\
&= c(\Theta) + \sum_{i \in \Theta} \big(f_M (U \setminus \{i\}) - f_M(U)\big)\\
&= c(\Theta) + \frac{(\gamma + 1)}{\gamma \cdot n}\mathlarger{\sum_{i \in \Theta}}\bigg[ \bigg|\bigcup\limits_{S \in U \setminus \{i\}}S\bigg| - \bigg|\bigcup\limits_{S \in U}S\bigg|\bigg]
\end{align*}
\normalsize
Note that all the hard instances of \textsf{SetCover} and \textsf{Max Coverage} are derived from the construction of \cite{LundY94}. All such instances are such that each element has multiple subsets which may cover it (intuitively if there is only one subset which covers a particular element in any hard instance, then we will pick it and get a smaller, easier instance of the problem). Since the union of all subsets of the given instance is $n$ and so is the union of all but one of the available subsets in the hard instance, each term in the above summation is 0. This implies that $c^*(\Theta) = c(\Theta)$ and we are done.\end{proof}
\vspace{-0.8em} 
\section{Speeding up the Marginal Greedy}\label{speedup}
In the worst case, the \margG algorithm runs in $\bigoh(n^2 \cdot \mbox{EO})$ time, where $n$ is the number of shareable nodes and EO is the time to evaluate $bc(S)$, i.e., the time to optimize the batch of queries given the set of nodes $S$, to be materialized. This makes the algorithm expensive since $n$ itself may be exponential in the worst case. Thus, we would like to reduce the time taken by the algorithm without sacrificing on the theoretical guarantees on the quality of the solution proved in Section \ref{NewGreedy}. In this section, we present some optimizations to our algorithm to improve its running time.
\subsection{Basic Optimizations}
We first note that two optimizations presented in \cite{Roy2000} can be used for our algorithm as well. Their first observation is about searching only over all the shareable nodes. As noted above, this can be directly used by us since our algorithm just presents a different heuristic for choosing which nodes to materialize. Their second optimization presents a way to incrementally update the \textit{bestCost} function for various sets that exploits the result of earlier cost computations to incrementally compute the new plan. Since the $mb$ function is just a linear transformation of the \textit{bestCost} function and our greedy algorithm (at least when the decomposition presented in the proof of Proposition \ref{decomp} is used) is also concerned with just successive differences in the values of the \textit{bestCost} function, their optimization can also be used to speed up our algorithm; for details see \cite{Roy2000}. 

Another optimization (not in \cite{Roy2000}) that can be made is based on a simple observation of the greedy algorithm and by exploiting submodularity. In the $i^{th}$ iteration, the \margG algorithm needs to compute the maximum benefit to cost ratio $\frac{f'_M(e,X_{i-1})}{c(\{e\})}$. Thus, if while scanning elements to compute the maximum, we encounter an element which has the marginal-benefit to cost ratio less than 1, we can remove it from the set $Y$ of elements to be searched over as it will never be picked by the \margG algorithm in the future iterations either. This is because $f_M$ is also submodular and the size of $X_i$ always increases as $i$ increases so the value of the marginal-benefit to cost ratio only decreases as the algorithm proceeds and will never become greater than 1. A similar optimization for the simple greedy algorithm used for monotone, submodular maximization under cardinality constraints is also possible.
 \subsection{The \textsc{LazyMarginalGreedy} algorithm}
 The third optimization in \cite{Roy2000} essentially leverages supermodularity to improve the running time of the greedy algorithm. The argument is similar to that used by \cite{Minoux} for the \textsc{LazyGreedy} algorithm. We observe that a similar argument as the ones presented in these two papers may be used for the \margG algorithm and is presented next.

 As noted previously, in each iteration $i$, the \margG algorithm must identify the element $e$ with the maximum marginal-benefit to cost ratio $\frac{f'_M(e,X_{i-1})}{c(\{e\})}$. For each element $e$, the denominator is fixed and the marginal benefits are monotonically nonincreasing during the iterations of the algorithm, i.e., $f'_M(e,X_i) \geq f'_M(e,X_j)$ whenever $i \leq j$. Thus, instead of recomputing $\frac{f'_M(e,X_{i-1})}{c(\{e\})}$ for each element $e \in V$, which requires $\bigoh(n)$ computations of $f$, the \textsc{LazyMarginalGreedy} algorithm maintains a list of upper bounds $u(e)$ (initialized to a  large value) on the marginal-benefit to cost ratio sorted in decreasing order (using a heap). 

 In each iteration, the algorithm extracts the element with largest $u(e)$ from the ordered list of remaining elements. If, after this update, $u(e) \geq u(e') \ \forall e' \neq e$, then submodularity guarantees that $\frac{f'_M(e,X_{i-1})}{c(\{e\})} \geq \frac{f'_M(e',X_{i-1})}{c(\{e\})} \ \forall e' \neq e$, and therefore the algorithm has identified the element with the largest marginal benefit to cost ratio without  computing the ratio for a potentially large number of elements $e'$. 
\subsection{Universe Reduction under size constraints}
We may sometimes want to consider a cardinality constraint (say $k$) on the number of nodes to be materialized. This may arise due to storage constraints which only allow materialization of a few subexpressions. We adapt our greedy algorithm for this constraint by simply stopping after $k$ elements are picked. 

While, at this point, we do not show any theoretical approximation guarantees for this problem, there is a way to leverage this cardinality constraint to prune out certain elements from the ground set $U$. This preprocessing step may be used to reduce the size of the set of PQDAG nodes $U$ on which the algorithm will be run. 

We show that the algorithm run on this reduced set is the same as that obtained when the algorithm runs on the full set. This check is useful only when there is a cardinality constraint of $k < n$, as we will show.

\begin{theorem} Let $U = \{e_1, \ldots, e_n\}$ be the set of all shareable PQDAG nodes ordered as $\\ \frac{f'_M(e_1,U  \setminus  \{e_1\})}{c(\{e_1\})}  \geq  \ldots \geq \frac{f'_M(e_n,U \setminus \{e_n\})}{c(\{e_n\})}$. Furthermore, let\\
$U' = \{e \in U \big| \frac{f_M(e)}{c(\{e\})} \geq \frac{f'_M(e_k,U \setminus \{e_k\})}{c(\{e_k\})}\} \mbox{ for } k < n$.\\
The output of the \margG algorithm (with cardinality constraint of $k$) when it runs on $U$ is the same as the output when it runs on $U'$. 
\end{theorem}
\begin{proof} The proof is provided in Appendix A.
\end{proof}
\begin{figure*}
\centering
\begin{subfigure}[b]{0.31\linewidth}
\centering
\includegraphics[height=4.5cm,width=\linewidth]{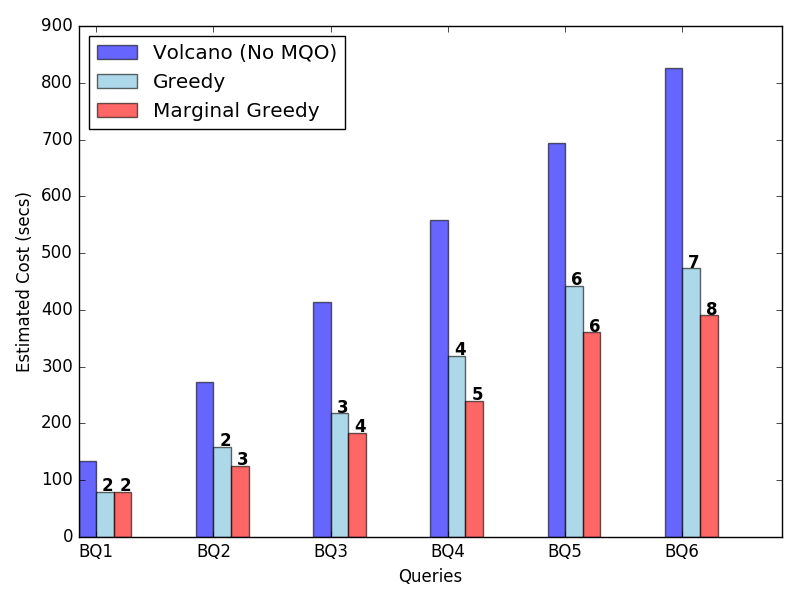}
\caption{1GB Total Size}
\label{fig:comp1gb}
\end{subfigure}
\begin{subfigure}[b]{0.31\linewidth}
\centering
\includegraphics[height=4.5cm,width=\linewidth]{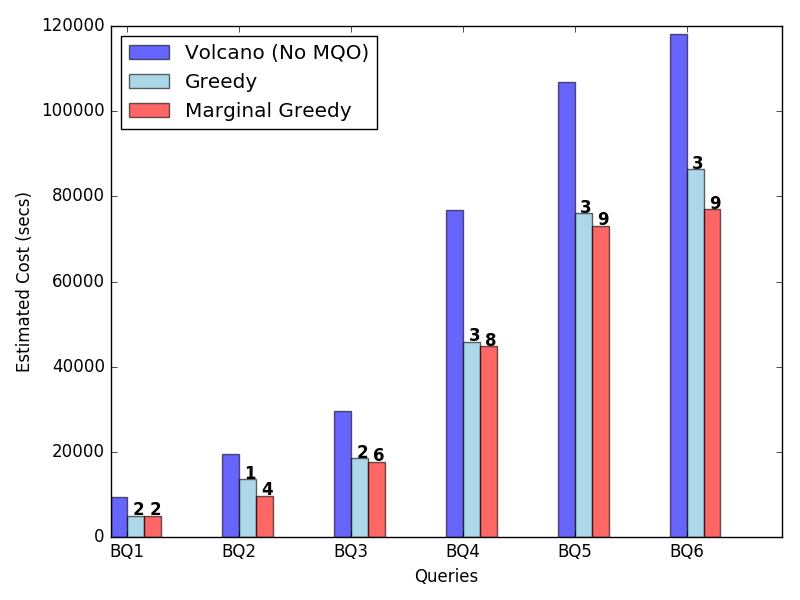}
\caption{100GB Total Size}
\label{fig:comp100gb}
\end{subfigure}
\begin{subfigure}[b]{0.31\linewidth}
\centering
\includegraphics[height=4.5cm,width=\linewidth]{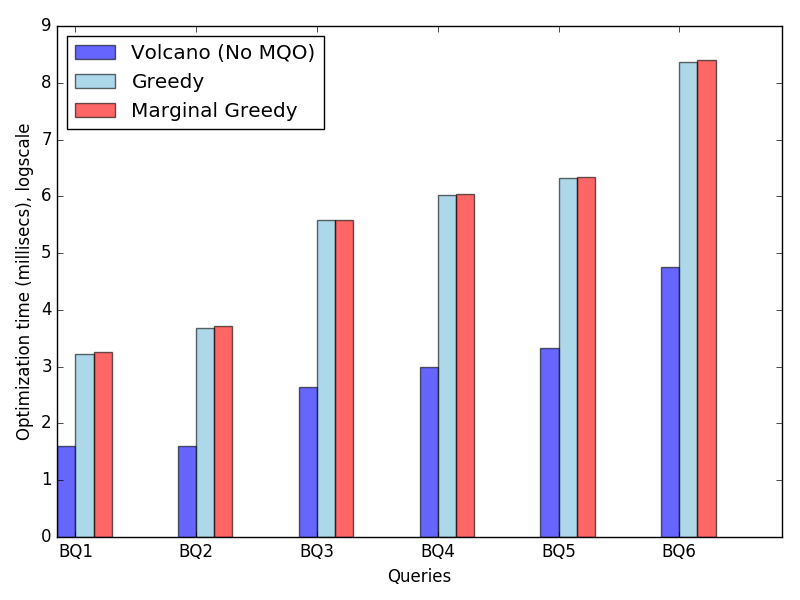}
\caption{Optimization Time (logscale)}
\label{fig:comptime}
\end{subfigure}
\caption{Results for batched TPCD queries (Experiment 1)}
\label{fig:tpcdcomp}
\end{figure*}

It is important to note that this strategy may not always lead to a reduction in the ground set but it may lead to pruning in certain cases.

Note that this pruning procedure can be modified to work for the simple greedy algorithm for monotone, submodular maximization under cardinality constraints. The proof is also along similar lines as those stated above.
\vspace{-0.8em}
\section{Experimental Section}\label{exptSection}
We now describe our experimental setup and findings. We worked with the original C++ code of Pyro which implemented the Greedy algorithm \cite{Roy2000}. We extended it by implementing the Marginal Greedy algorithm. All the optimizations discussed in Section 5 are implemented with the exception of the one discussed in subsection 5.3 as we are mainly interested in the best plan without imposing any cardinality constraints.

The optimizer rule set consists of select push down, join commutativity and associativity (to generate bushy join trees), and select and aggregate subsumption. The physical operators included sort-based aggregation, merge join, nested loop join, indexed selection and relation scan. The implementation includes handling physical properties (sort order and presence of indices) on base and intermediate relations, unification and subsumption during DAG generation (see \cite{Roy2000} for details).

 The block size was taken as 4KB and our cost functions assume 6MB is available to each operator during execution (we also conducted experiments with memory sizes of 128MB). Standard techniques were used for estimating costs, using statistics about relations. The cost estimates are of the standard resource consumption estimates which contain an I/O component and a CPU component, with seek time as 10 msec, transfer time of 2 msec/block for read and 4 msec/block for write, and CPU cost of 0.2 msec/block of data processed. 

We assume that intermediate results are pipelined to the next input, using an iterator model as in Volcano; they are saved to disk only if the result is to be materialized for sharing. The materialization cost is the cost of writing out the results sequentially. The tests were performed on a 2.4 GHz Intel i7 processor laptop with 8GB memory running Linux. We compare Marginal Greedy with Greedy and stand-alone Volcano (no MQO). The optimization time of our Marginal Greedy algorithm was very close to that of the Greedy algorithm in \cite{Roy2000}. The optimization times are measured as CPU time.
\vspace{-0.5em}
\subsection{Experiment 1 (Batched TPCD Queries)}
The workload for the first experiment models a system where several TPCD queries are executed as a batch. The workload consists of subsequences of the queries Q3, Q5, Q7, Q8, Q9 and Q10. Each query was repeated twice with different selection constants. Composite query BQ$i$ consists of the first $i$ of the above queries, and we used composite queries BQ$1$ to BQ$6$ in our experiments. The TPCD database is used at a scale of 1 (1 GB total size), with a clustered index on the primary keys for all the base relations. We also ran the queries in this experiment and the next at a scale of 100 (total size 100GB).

Note that although a query is repeated with two different values for a selection constant, we found that the selection operator generally lands up at the bottom of the best Volcano plan tree, and the two best plan trees may not have common subexpressions.

The results on the two workloads (1GB and 100 GB total sizes) are shown in Figure \ref{fig:tpcdcomp}. The number on top of the bars for Greedy and Marginal Greedy denotes the number of materialized nodes. Greedy does substantially better than Volcano (without MQO) by upto 57\%. Marginal Greedy always does as good as or better than Greedy. In fact, the results are the same only for BQ$1$ where both chose to materialize the two nodes which lead to benefit. For all other queries in the experiment with 1GB Total Size, the improvement of Marginal Greedy is always between 12\% and 25\%. This is primarily due to the number of materialized nodes by Marginal Greedy being more than that by Greedy. BQ$5$ is especially interesting in Figure \ref{fig:comp1gb} as the number of materialized nodes is the same yet there is almost a 20\% improvement over Greedy. In fact, for queries from BQ$4$ to BQ$6$, the intersection in the materialized nodes by the two algorithms had an overlap of 1 or 2 only.

In the experiment with 100GB Total Size (Figure \ref{fig:comp100gb}), as mentioned, the nodes chosen to be materialized for BQ$1$ are the same for both algorithms. For the rest of the queries, the number of materialized nodes is much larger than in the 1GB size dataset. While the relative gains in this dataset might seem comparable or slightly lesser than those observed in the smaller dataset, the actual gains in these cases are substantial due to large costs coming from these large data sizes. In these queries, there were 1 or 2 nodes which had substantially more benefit and got picked by both Greedy and Marginal Greedy. While Greedy picked a few more nodes which seemed benefical initially, Marginal Greedy picked many more nodes, each of which had moderate benefit but lead to an overall decrease in the cost. This behaviour was particularly observed in BQ$5$ and BQ$6$ and we believe for larger sets of queries on larger data sets, this behavior will be more pronounced.

The optimization times for the queries are shown in Figure \ref{fig:comptime}. Since the values for Greedy and Marginal Greedy were very close to each other, we present the results in logscale. As can be seen, even in logscale, the optimization times are very close to each other. We stress that while the execution cost of a query depends on the size of the underlying data, the cost of optimization does not.
\vspace{-0.6em}
\subsection{Experiment 2 (Stand-Alone TPCD Queries)}
Roy et al. \cite{Roy2000} also had an experiment consisting of four individual queries based on TPCD using the same data sizes (1GB and 100GB) and the same indices. These queries had common subexpressions within themselves and benefitted from MQO to optimize just those queries individually. However, in all four queries, only node was beneficial and hence, both algorithms found that node and resulted in the same answer. We present the results here for completeness. We explain these queries themselves and the actual results are presented in Figure \ref{fig:tpcd} in Appendix B.

TPCD query Q2 has a large nested query, and repeated invocations of the nested query in a correlated evaluation could benefit from reusing some of the intermediate results. Greedy and Marginal Greedy gave a plan with an estimated cost of 79 seconds for the smaller data set and 1929 seconds for the larger one. Decorrelation is an alternative to correlated evaluation and Q2-D is a (manually) decorrelated version of Q2 (due to decorrelation, Q2-D is actually a batch of queries). Multi-query optimization also gives substantial gains on the decorrelated query Q2-D, results in a plan of estimated cost 46 and 2059 for the two data sizes respectively, by both algorithms.
We next considered the TPCD queries Q11 and Q15, both of which have common subexpressions, and hence make a case for multi-query optimization. For Q11, both the greedy algorithms lead to a plan of approximately half the cost as that returned by Volcano. The improvements for Q15 are similar but more pronounced for the smaller data set. 

The conclusion based on the experiments seems to be that when there are multiple possible nodes that can be materialized, Greedy chooses the nodes which result in considerable improvements early on but Marginal Greedy is more global and chooses to materialize more nodes which might have moderate benefit individually but can result in overall benefits.

\vspace{-0.8em}
\section{Related Work}\label{RelWork}
We now present the related work in the areas of multi-query optimization and submodular maximization.
\subsection{Multi-Query Optimization}
The MQO problem has received significant attention in the past \cite{Sellis:1988:MO:42201.42203,Park:1988:UCS:645473.653403,Shim1994,Rosenthal1988,Subramanian1998}. Initial work \cite{Sellis:1988:MO:42201.42203,Park:1988:UCS:645473.653403,Rosenthal1988,Shim1994} proposed solutions that were not fully integrated with the query optimizer and were primarily exhaustive. 

Subramanian and Venkataraman \cite{Subramanian1998} consider sharing only among the best plans of the query; this approach can be implemented as an efficient, post-optimization phase in existing systems, but can be highly suboptimal. 

To choose the set of nodes to be materialized, Roy et al. \cite{Roy2000} use a greedy algorithm which has already been discussed in detail in Section \ref{Prelims}. Dalvi et al. \cite{Dalvi2003} explores the possibility of sharing intermediate results by pipelining, avoiding unnecessary materializiations. Diwan et al. \cite{Sudarshan2006} consider the MQO problem in Volcano taking scheduling and caching into account. They present an exhaustive algorithm which takes $\bigoh(n^n)$ time, which is clearly infeasible. 

 Zhou et al.\cite{Zhou2007} propose a framework to use common subexpressions for MQO and materialized view selection in a query optimizer based on the Cascades framework \cite{Graefe1995}. The focus however is on ``covering'' subexpressions at the LQDAG level and they do not take into account competing physical properties like sort orders and partitioning properties from different consumers. 

 Silva et al. \cite{Silva2012} consider physical properties in a cost-based fashion. However, their solution is also based on heuristics which materializes \textit{every} common subexpression at the LQDAG level. The best physical property for each subexpression is chosen and all consumers are forced to use the same physical property, which can be sub-optimal. Furthermore, even with this heuristic, their approach can be very expensive when there are many potential physical properties for each subexpression.

\subsection{Submodular Maximization}
Submodular maximization has received a significant amount of attention in optimization \cite{Calinescu:2011:MMS:2340436.2340447,Nemhauser1978,Buchbinder2012} and has wide applicability in machine learning, computer vision and information retrieval \cite{JegelkaB2011,Kempe2003,Boykov2001,Torr}. In this problem, we are given a submodular function $f$ and a universe $U$, with the goal of selecting a subset $S \subseteq U$ such that $f(S)$ is maximized. Typically, $S$ must satisfy additional feasibility constraints such as cardinality, knapsack or matroid constraints.

This problem is $\np$-hard even for the simplest problems which involve only \textit{cardinality constraints} and monotone functions. Nemhauser et al. \cite{Nemhauser1978} show that a simple greedy algorithm gives a $(1-1/e)$ approximation for monotone submodular maximization under cardinality constraints. They further show that it is $\np$-hard to obtain a better approximation guarantee. Sviridenko \cite{DBLP:journals/orl/Sviridenko04} presents a modified greedy algorithm for monotone submodular function maximization under knapsack constraints. Their algorithm is the main motivation for our marginal greedy algorithm. 

Buchbinder et al. \cite{Buchbinder2012} gave a $1/2$-approximation algorithm for unconstrained non-monotone submodular maximization, for which there is a matching hardness result. However, all these results assume non-negativity of the function $f$. Mittal and Shulz \cite{MittalS13} show that a constant factor approximation for non-negative supermodular minimization is $\np$-hard. Inapproximability of non-monotone submodular maximization (with possibly negative values) is also well known. To the best of our knowledge, ours is the first work which, under the assumption of $f(\emptyset) = 0$, provides an approximation algorithm with a matching hardness of approximation result for unconstrained non-monotone submodular maximization when the function may take negative values. Since the hardness of approximation factor depends on the optimal (and may go to 0), this rules out constant factor approximations for the problem even in the restricted setting of $f(\emptyset)=0$. 
\vspace{-0.8em}
\section{Conclusions and Future Work}\label{Conclusion}
In this paper, we have presented a reformulation of the well-studied MQO problem. Under the assumption of supermodularity of the \textit{bestCost} function, we propose a greedy algorithm for the maximization problem and provide an approximation factor guarantee for our algorithm. We then showed that obtaining a better approximation factor than the one attained by our greedy algorithm is $\np$-hard. Such theoretical guarantees on the quality of any heuristic has not been presented before. Since the underlying problem solved in this paper is the unconstrained, normalized submodular maximization problem, with possibly negative values, we believe our results can be useful beyond just MQO.

One area of future work is the problem of \textit{non-negative}, non-monotone submodular maximization problem under cardinality constraints and more generally, matroid constraints. This is an outstanding open problem and even the most recent work \cite{EneMatroidSubmodularMax} has a considerable gap in the approximation ratio and the hardness of approximation known. We would like to see if ideas in this paper like the ``best decomposition'' can be used to devise algorithms with better guarantees for that problem. 
\section{Acknowledgments}
We would like to thank Amit Deshpande, Deeparnab Chakrabarty, Ravishankar Krishnaswamy and Sebastien Tavenas for comments on the paper.
%

\bibliographystyle{abbrv}
\bibliography{vldb_sample}

%
%
 \appendix
 \section{Additional Proofs}
 We now present the missing proofs. We reproduce the theorem statements for convenience.

\vspace{0.2em}
\textbf{Proposition} 1. (in the main paper) Any normalized, non-monotone (which may take negative values) submodular function $f$ can be decomposed as \begin{align*} f(S) = f_M(S) - c(S) \ \ , \forall \ S \subseteq U\end{align*} where $f_M$ is a monotone submodular function and $c$ is an additive cost function. In particular, one possible decomposition is
\begin{align*}
& f_M^*(S) = f(S) + \sum\limits_{e \in S}(f(U \setminus \{e\})-f(U))\\
& c^*(S) = \sum\limits_{e \in S}(f(U \setminus \{e\}) - f(U)) 
\end{align*}

\begin{proof} It is easy to see that $c$ is additive and that \begin{align*} \forall \ S \subseteq X, \mbox{ we have } f(S) = f_M(S) - c(S) \end{align*}
Since $c$ is additive and $f$ is submodular, $f_M$ is also submodular since for arbitrary $S_1 \subset S_2 \subset U$ and an arbitrary $ e \in U \setminus S_2$, 
\begin{align*} & f_M(S_1 \cup \{e\}) - f_M(S_1)&\\
= \ & f(S_1 \cup \{e\}) - c(S_1 \cup \{e\}) - f(S_1) + c(S_1) &\\
= \ & f(S_1 \cup \{e\}) - f(S_1) - c(\{e\}) \ \  (\mbox{by linearity of } c)&  \\
\geq \ &f(S_2 \cup \{e\}) - f(S_2) - c(\{e\}) \ \ (\mbox{by submodularity of } f)&\\
= \ & f(S_2 \cup \{e\}) - f(S_2) - c( S_2 \cup \{e\}) + c(S_2) \ \  (\mbox{by linearity})&\\
= \ & f_M(S_2 \cup \{e\}) - f_M(S_2). &
\end{align*} 
Now we just have to show that $f_M$ is monotone. \\
Consider an arbitrary $S \subset U$ and an arbitrary $ e \in U  \setminus  S$. Let us consider the expression 
\begin{align*} & \ f_M(S\cup\{e\}) -  f_M(S)\\ 
= \ &  \  f(S\cup\{e\}) - f (S) + (f(U  \setminus  \{e\}) - f(U))\\
= \ &  \  ( f(S\cup\{e\}) - f (S) ) - (f(U) -  f(U  \setminus  \{e\}) )\\
\geq \  &  0 & \end{align*}
The inequality in the last line follows from the fact that $S \subseteq U  \setminus \{e\}$ and the submodularity of $f$. The terms in the summation can be suitably scaled to ensure that $c$ is zero only at $\emptyset$ and positive everywhere else.
\end{proof}

\textbf{Proposition 2}. (in the main paper) Given an arbitrary decomposition $f_M$ and $c$ of a normalized submodular function $f$, i.e., $f(S) = f_M(S) - c(S) \ \forall \ S \subseteq V$ with monotone $f_M$ and consider another decomposition 

\begin{align*}
\widetilde{f}_M (S) &= f_M (S) - \sum_{i \in S} \big(f_M (U) - f_M (U \setminus i) \big) \\
\widetilde{c} (S) &= c(S) - \sum_{i \in S} \big(f_M (U) - f_M (U \setminus i) \big)
\end{align*}

Then, $\widetilde{f}_M$ is monotone. Furthermore, for the decomposition in Proposition \ref{decomp}, $f^*_M$ and $c^*$, $\widetilde{f}^*_M = f^*_M$ and $\widetilde{c}^* = c^*$. 
\begin{figure*}[t]
\centering
\begin{subfigure}[b]{0.31\linewidth}
\centering
\includegraphics[height=4.5cm,width=\linewidth]{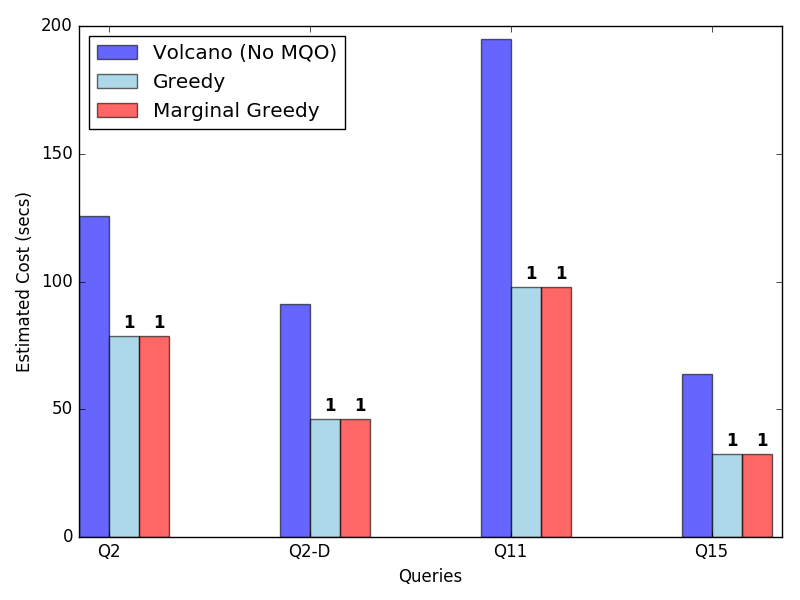}
\caption{1GB Total Size}
\label{fig:tpcd1gb}
\end{subfigure}
\begin{subfigure}[b]{0.31\linewidth}
\centering
\includegraphics[height=4.5cm,width=\linewidth]{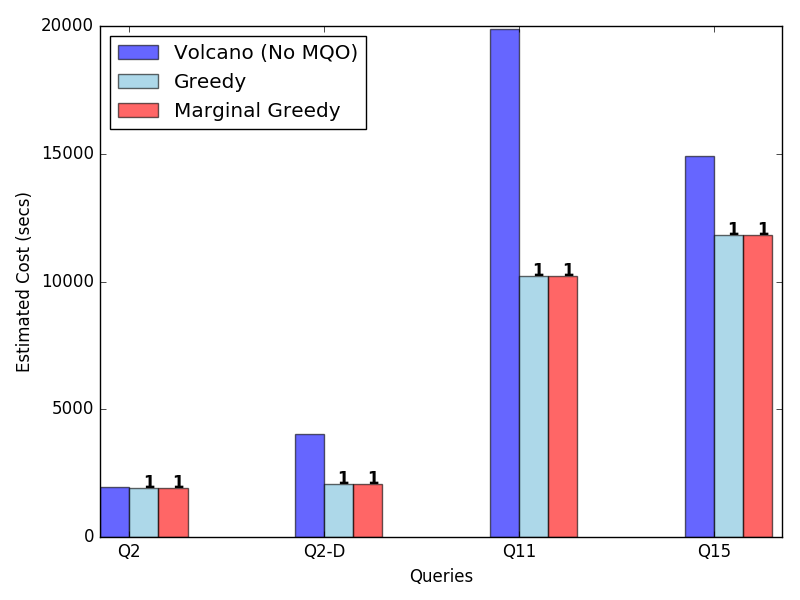}
\caption{100GB Total Size}
\label{fig:tpcd100gb}
\end{subfigure}
\begin{subfigure}[b]{0.31\linewidth}
\centering
\includegraphics[height=4.5cm,width=\linewidth]{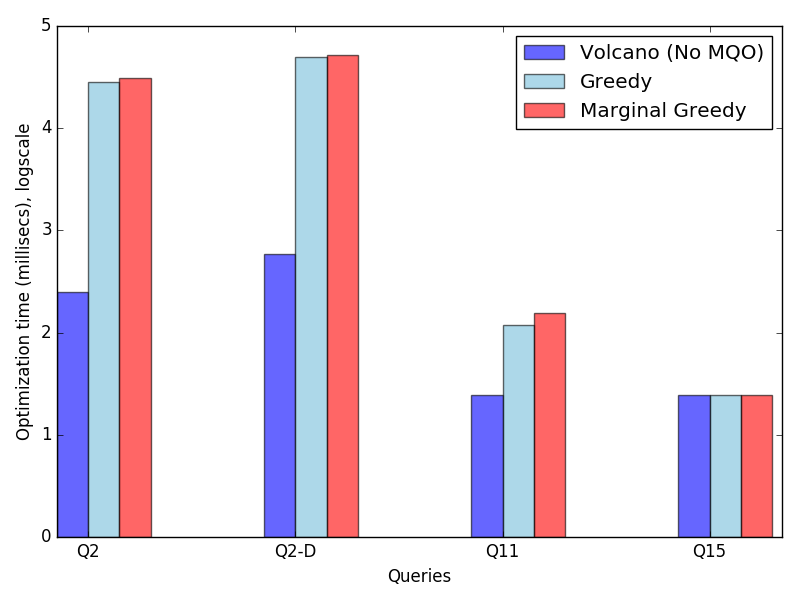}
\caption{Optimization Times (logscale)}
\label{fig:tpcdtimes}
\end{subfigure}
\caption{Results for stand-alone TPCD queries (Experiment 2)}
\label{fig:tpcd}
\end{figure*}
\begin{proof} To show monotonicity of $\widetilde{f}_M$, it is enough to show $\forall j \in U, \ \forall S \subseteq U \setminus \{j\}, \widetilde{f}_M(S \cup \{j\}) - \widetilde{f}_M (S) \geq 0$. 

\begin{align*}
&\widetilde{f}_M(S \cup \{j\}) - \widetilde{f}_M (S)\\
&= f_M(S \cup \{j\}) - f_M (S) -  \big(f_M (U) - f_M (U \setminus \{j\}) \big)\\
&\geq 0 \ \ \ \mbox{(by submodularity of } f_M)
\end{align*}

For the second part, we just expand the expressions to get the desired result.
\small
\begin{align*}
\widetilde{c}^* (S) &= c^* (S) - \sum\limits_{i \in S} \big(f_M^*(U) - f_M^*(U \setminus\{i\}) \big)\\
&= c^* (S) - \sum\limits_{i \in S} f(U) - f(U \setminus\{i\}) + \big(f(U \setminus \{i\}) - f(U) \big)\\
&= c^* (S)
\end{align*}
\normalsize
The computation for $\widetilde{f}_M^*(S)$ proceeds similarly.
\end{proof}
\vspace{0.2em}
 \textbf{Theorem} 4. (in the main paper) Let the set of all shareable PQDAG nodes $U = \{e_1, \ldots, e_n\}$ be ordered as $\\ \frac{f'_M(e_1,U  \setminus  \{e_1\})}{c(\{e_1\})}  \geq  \ldots \geq \frac{f'_M(e_n,U \setminus \{e_n\})}{c(\{e_n\})}$. Furthermore, let\\
$U' = \{e \in U \big| \frac{f_M(e)}{c(\{e\})} \geq \frac{f'_M(e_k,U \setminus \{e_k\})}{c(\{e_k\})}\} \mbox{ for } k < n$.\\
The output of the \margG algorithm (with cardinality constraint of $k$) when it runs on $U$ is the same as the output when it runs on $U'$. 

\begin{proof}
Without loss of generality, assume that the algorithm, when run on V, terminates after the full $k$ steps. Let the sequence of chosen elements, in order of inclusion, be $\{s_1, s_2, \ldots, s_k\}$ and for all $i \in [k]$, let $X_i = \{s_1, s_2, \ldots, s_i\}$, as before. Clearly, $\emptyset = X_0 \subset X_1 \subset X_2 \subset \ldots \subset X_k$.

\textbf{Case 1.} $k = n$\\
This is a simple case in which all elements are chosen and, thus, $U'$ should be equal to $U$ which is shown as follows
 $\forall e \in U$, we have

\begin{align*}
   \frac{f_M(\{e\})}{c(\{e\})} &= \frac{f'_M(e,\emptyset)}{c(\{e\})} \\
  & \geq \frac{f'_M(e, U \setminus \{e\})}{c(\{e\})} & \mbox{(by submodularity)}\\
  & \geq \frac{f'_M(e_k, U \setminus \{e_k\})}{c(\{e_k\})}.
\end{align*}

Hence, all elements of $U$ are going to be in $U'$ since they all satisfy the condition to be in $U'$. In this case, the check is clearly wasteful since the ground set has no reduction and a lot of functional calls are made. In the MQO context, this corresponds to invoking a lot of \textit{bestCost}$(Q, S)$ calls, each of which are moderately expensive. Thus, in this case, the preprocessing step should just check if $k = n$ and if so, directly pass the full ground set to the greedy algorithm.

\textbf{Case 2.} $ k < n \ \& \ X_k = \{e_1, e_2, \ldots, e_k\}$. \\In this case, the theorem follows trivially since $U'$ will contain all elements in $X_k$, along with some other elements. 

\textbf{Case 3.} $ k < n \ \& \ X_k \neq \{e_1, e_2, \ldots, e_k\}$.\\
We first make a claim which we will prove later.
\begin{claim}\label{intermediateClaim} For all $i \in \{1, 2,\ldots, k\}$, we have
\begin{align*}
 \frac{f'_M(s_i, X_{i-1})}{c(\{s_i\})} \geq \frac{f'_M(e_i, U \setminus \{e_i\})}{c(\{e_i\})}.
\end{align*}
\end{claim}
The claim is used to show that elements in $U \setminus U'$ will never be picked by the \margG algorithm. Intuitively, for any element $e \notin U'$, the $\frac{f'_M(e,X_i)}{c(\{e\})}$ ratio of picking it is largest in the first iteration (by submodularity) and that itself is less than the element with the smallest ratio of the elements selected by the greedy algorithm. So, it is guaranteed that the greedy algorithm does not pick any element which is not in $U'$. This is easy to see and is as follows\\
For all $ e \in U \setminus U'$, we have 
\begin{align*}
 &\frac{f'_M(e,\emptyset)}{c(\{e\})} = \frac{f_M(e)}{c(\{e\})} < \frac{f'_M(e_k,U \setminus \{e_k\})}{c(\{e_k\})}.
 \end{align*}
 By Claim \ref{intermediateClaim},
\begin{align*}
 & \frac{f'_M(s_k,X_{k-1})}{c(\{s_k\})} \geq \frac{f'_M(e_k,U \setminus \{e_k\})}{c(\{e_k\})}  \\
\implies & \frac{f'_M(s_k,X_{k-1})}{c(\{s_k\})} > \frac{f'_M(e,\emptyset)}{c(\{e\})}.
\end{align*}
We now present the proof of Claim \ref{intermediateClaim}.
\begin{proof}(of Claim \ref{intermediateClaim}) The case of $e_i \not\in X_{i-1}$ is trivial due to submodularity and the greedy algorithm.

Thus, we just have to prove for the case when $e_i \in X_{i-1}$\\
Since $|X_{i-1}| = i - 1, X_{i-1}$ cannot include all elements from the set $\{e_1, e_2, \ldots, e_i\}$. Thus, there exists some element, say, $e_z \in {e_1, e_2,\ldots, e_i}$ such that $e_z \notin X_{i-1}$.

\begin{align*}
\mbox{Thus, we have } & \hspace{0.5em} \frac{f'_M(s_i,X_{i-1})}{c(\{s_i\})} \\
=&  \max\limits_{e \in U \setminus X_{i-1}} \frac{f'_M(e, X_{i-1})}{c(\{e\})}\\
\geq &  \frac{f'_M(e_z, X_{i-1})}{c(\{e_z\})} \\
\geq &  \frac{f'_M(e_z, U \setminus \{e_z\})}{c(\{e_z\})} \mbox{ (by submodularity)} \\
\geq &  \frac{f'_M(e_i, U \setminus \{e_i\})}{c(\{e_i\})}
\end{align*}\end{proof}
This concludes our proof.
\end{proof}
\section{Results of Experiment 2}
In this section, we present the results of Experiment 2 (Stand-alone TPCD). The results are shown in Figure \ref{fig:tpcd}.

\end{document}